\documentclass[11pt,letterpaper]{article}
\usepackage{times}
\usepackage{amsmath,amsfonts,graphpap,amscd,dsfont,mathrsfs,graphicx,lscape}
\usepackage{epsfig,amssymb,amstext,xspace}
\usepackage{theorem}
\usepackage{algorithm,algorithmic}

\usepackage{color}              
\usepackage{epsfig}

\newcommand{\E}{\mathbb{E}}

\newtheorem{theorem}{Theorem}
\newtheorem{lemma}[theorem]{Lemma}

{\theorembodyfont{\rmfamily} }

\def\squareforqed{\hbox{\rule{2.5mm}{2.5mm}}}

\def\QED{\ifmmode\squareforqed 
  \else{\nobreak\hfil   
    \penalty50                 
    \hskip1em                  
    \null                      
    \nobreak                   
    \hfil                      
    \squareforqed              
    \parfillskip=0pt           
    \finalhyphendemerits=0     
    \endgraf}                  
  \fi}

\def\blksquare{\rule{2mm}{2mm}}
\def\qedsymbol{\blksquare}
\newcommand{\bg}[1]{\medskip\noindent{\bf #1}}
\newcommand{\ed}{{\hfill\qedsymbol}\medskip}
\newenvironment{proof}{\noindent\textbf{Proof.}}{\QED}

\newenvironment{proofof}[1]{\bg{Proof of #1 : }}{\ed}

\newcommand{\R}{\ensuremath{\mathbb R}}

\newcommand{\comment}[1]{}

\newcommand{\junk}[1]{}

\newcommand{\prob}{\mathbb{P}}

\newlength{\tmp} \newlength{\lpsx} \newlength{\lpsy} \newlength{\upsx} \newlength{\upsy}

\newcommand{\bid}{b}
\newcommand{\bids}{{\mathbf \bid}}
\newcommand{\bidsmi}{{\mathbf \bid}_{-i}}
\newcommand{\bidi}[1][i]{{\bid_{#1}}}

\newcommand{\val}{v}
\newcommand{\vals}{{\mathbf \val}}
\newcommand{\valsmi}{{\mathbf \val}_{-i}}
\newcommand{\vali}[1][i]{{\val_{#1}}}

\newcommand{\util}{u}

\newcommand{\utili}[1][i]{{\util_{#1}}}

\newcommand{\dist}{F}
\newcommand{\dists}{{\mathbf \dist}}
\newcommand{\distsmi}{\dists_{-i}}

\newcommand{\price}{p}

\newcommand{\pricei}[1][i]{{\price_{#1}}}

\title{Improved Social Welfare Bounds for GSP at Equilibrium}

\author{Brendan Lucier\thanks{\texttt{blucier@cs.toronto.edu}. Dept. of Computer Science, University of Toronto, Toronto, ON.} \and Renato Paes Leme\thanks{\texttt{renatoppl@cs.cornell.edu}. Dept. of Computer Science, Cornell University, Ithaca, NY.}}

\date{}
\begin{document}

\maketitle

\abstract{
The Generalized Second Price auction is the primary method by which sponsered search advertisements are sold.  We study the performance of this auction under various equilibrium concepts.  In particular, we demonstrate that the Bayesian Price of Anarchy is at most $2(1-1/e)^{-1} \approx 3.16$, significantly improving upon previously known bounds.

Our techniques are intuitively straightforward and extend in a number of ways.  For one, our result extends to a bound on the performance of GSP at coarse correlated equilibria, which captures (for example) a repeated-auction setting in which agents apply regret-minimizing bidding strategies.  
In addition, our analysis is robust against the presence of byzantine agents who cannot be assumed to participate rationally.

Additionally, we present tight bounds for the social welfare obtained at pure NE for the special case of an auction for 3 slots, and discuss potential methods for extending this analysis to an arbitrary number of slots.
}

\section{Introduction}

The sale of advertising space is the primary source of revenue for many providers of online services.  This is due, in part, to the fact that providers can tailor advertisements to the preferences of individual users.  A search engine, for example, can choose to display ads that synergize well with a query being searched.  However, such dynamic provision of content complicates the process of selling ad space to potential advertisers.
The now-standard method has advertisers place bids -- representing the amount they would be willing to pay per click -- which are resolved in an automated auction whenever ads are to be displayed.

By far the most popular bid-resolution method currently in use is the Generalized Second Price auction (GSP), a generalization of the well-known Vickrey auction.  In the GSP, there are multiple ad ``slots'' of varying appeal (i.e.\ slots at the top of the page are more effective).  Advertisers are assigned slots in order of their bids, with the highest bidders receiving the best slots; each advertiser then pays an amount equal to the bid of the next-highest bidder.  While simple to understand and use, the GSP has some notable drawbacks: unlike the Vickrey auction it is not truthful, and it does not generally guarantee the most efficient outcome (i.e.\ the outcome that maximizes social welfare).  Nevertheless, the use of GSP has been extremely successful in practice.  This begs the question: \emph{are there theoretical properties of the Generalized Second Price auction that would explain its prevalence?}

Here we continue the line of work aimed at answering this question by analyzing the performance of GSP under various models of rational agent behaviour.  
First, we consider Bayes-Nash equilibria (BNE) of GSP.  In this model, the auction is viewed as a partial-information game in which each participant's value per click is private information drawn independently at random from commonly-known distributions.  Such a model is particularly relevant for online ad auctions, since historical data can readily be observed to develop accurate market statistics.  A BNE is then a profile of bidding strategies whereby each agent maximizes his expected profit subject to the distribution over the other agents' values.  We study the expected social welfare that GSP attains at any BNE, as a fraction of the optimal social welfare.  This metric is commonly known as the Bayesian Price of Anarchy, representing the loss in efficiency due to having outcomes determined at BNE rather than a benevolent optimizer.

The BNE solution concept captures scenarios in which a large market of advertisers settle into a stable pattern of bidding strategies.  However, empirical studies show that bidding need not stabilize in some cases.  Advertising slot auctions can be repeated millions of times per day, and there are bidding patterns in which agents modify their strategies over time to respond to each others' bids.  To address such cases, one must consider GSP in the broader context of a repeated auction.  In such settings, we assume that an agent's value per click does not change over time, but declared bids can change each round. A solution concept then describes rational behaviour over many instances of the auction (i.e.\ a method of responding to the past play of other agents), and the metric of interest is the average social welfare attained by GSP over many rounds.

We consider an equilibrium model suited to long-run bidder behaviour in GSP.  Namely, we consider settings in which agents choose their bids so that their \emph{regret} vanishes over time.  Roughly speaking, such a model assumes that agents observe the bidding patterns of others and modify their own bids in such a way that their long-term performance approaches that of a single optimal strategy chosen in hindsight.  It is well-known from learning theory that such regret minimization is easy to achieve via simple bidding techniques.  We bound the Price of Total Anarchy, which is the ratio between the social welfare of the optimal allocation and the average social welfare obtained by GSP when agents minimize regret over a sufficiently long number of rounds.


\paragraph{Results}

Our main result is a bound on the social welfare obtained at Bayes-Nash equilbrium for the GSP auction.  Specifically, we show that the Bayesian Price of Anarchy for GSP is at most $2(1-1/e)^{-1} \approx 3.164$.  This improves upon the previous best-known bounds of 8 for BNE and 4 for (mixed) NE \cite{PLT10}. 

Perhaps just as important as the improved bounds, however, is the straightforward and robust nature of our proof.  In particular, our results extend to give the same bound for coarse correlated equilibria, which implies that the Price of Total Anarchy is at most $2(1-1/e)^{-1}$.  
%
Moreover, these results are resilient against the presence of Byzantine agents, in the following sense.  Suppose that, in addition to the rational participants in the auction, there is also some set of agents who apply arbitrary strategies.  We can view these as irrational participants who do not understand how to bid strategically.  
Note that, in such a setting, it is not possible for an auction to guarantee a fraction of the social welfare obtainable from the irrational bidders; after all, a bidder with very large value may decide (irrationally) to bid $0$ and effectively not participate in the auction.  What we \emph{can} show, however, is that the presence of the irrational bidders does not interfere with the auction's ability to approximate the welfare obtainable from the rational bidders. That is, the ratio of the optimal social welfare \emph{of the rational bidders} to the total social welfare obtained at any BNE is at most $2(1-1/e)^{-1} \approx 3.164$.  This result requires an assumption on the play of the irrational bidders, which is that no player bids more than his true value.  We feel that this is a reasonable assumption, as overbidding is a dominated strategy that is easily avoided; we therefore view the irrational bidders as novice or uninformed participants who would avoid dominated strategies, rather than truly adversarial agents.

Our results hold for a standard model of separable click-through rates, where the probability that a user clicks on an advertisement $j$ in slot $i$ is of the form $\alpha_i \gamma_j$.  That is, it is a product of two separable components: one for the advertiser, and one for the slot.  For ease of exposition, we will focus on the special case that $\gamma_j = 1$ for all $j$.  However, we note that our results extend to the more general case of separable click-through rates.

\paragraph{Related work}
In recent years there has been a surge of work on algorithmic mechanism design for sponsored search, beginning with Mehta et al. \cite{jacmMehtaSVV07,focsMehtaSVV05}.  See the survey of Lahaie et al \cite{algorithmgame} for an overview of subsequent developments.  The GSP model applied in this manuscript is due to Edelman et al \cite{edelman07sellingbillions} and Varian \cite{Varian06positionauctions}.

The work most closely related to ours is that of Paes Leme and Tardos, who also study equilibria of GSP \cite{PLT10}.  They give upper bounds on the Price of Anarchy in pure, mixed, and Bayesian strategies; achieving bounds of $1.618$, $4$, and $8$, respectively.  Our main result is a simplification and strengthening of their results for the mixed and Bayesian cases, as well as an extension to different but related solution concepts.

Edelman et al \cite{edelman07sellingbillions} and Varian \cite{Varian06positionauctions} study Envy-free equilibria of GSP (a special case of Nash equilibrium) in the full information setting.  They demonstrate that such equilibria exist, and that all such equilibria are socially optimal.  Gomes and Sweeney \cite{gomes09} study the Generalized Second Price Auction in the Bayesian context. They show that, unlike the full information case, there may not exist symmetric or socially optimal equilibria in this model, and obtain sufficient conditions on click-through-rates that guarantee the existence of a symmetric and efficient equilibrium.  Lahaie \cite{lahaie} also considers the problem of bounding the social welfare obtained at equilibrium, but restricts attention to the special case that click-through-rate $\alpha_i$ decays exponentially along the slots with a factor of $\frac{1}{\delta}$.  Lahaie proves a price of anarchy of $\min\{\frac{1}{\delta}, 1 - \frac{1}{\delta} \}$.

Lucier and Borodin \cite{borodin_lucier} study the Bayesian price of anarchy for greedy combinatorial auctions.
They show via a type of smoothness argument (see \cite{roughgarden}) that a greedy $c$-approximation algorithm can be turned into a mechanism with Price of Anarchy $c+1$ - for pure and mixed Nash and for Bayes-Nash equilibria.
Lucier \cite{lucier_ics} considers repeated greedy auctions and studies the design of mechanisms with bounded price of total anarchy and price of sinking.
These results do not imply bounds for GSP, since it is not a combinatorial auction (and, in particular, GSP does not provide a bidding language expressive enough to implement their mechanisms).  However, the approach taken in our work is similar to the one that drives their results.


The study of regret-minimization goes back to the work of Hannan on repeated two-player games \cite{H-57}.  Kalai and Vempala \cite{KV-05} extend the work of Hannan to online optimization problems, and Kakade et al \cite{KKL-07} further extend to settings of approximate regret minimization.  Blum et al \cite{BHLR-08} apply regret-minimization to the study of inefficiency in repeated games, coining the phrase ``price of total anarchy'' for the worst-case ratio between the optimal objective value and the average objective value when agents minimize regret.


\section{Preliminaries}

We consider an auction with $n$ advertisers and $n$ slots\footnote{we handle
unequal numbers of slots and advertisers by adding
virtual slots with click-through-rate zero or virtual advertisers with zero value per
click.}.
An \emph{outcome} is an assignment of advertisers to slots.  An outcome can be viewed
as a permutation $\pi$ with $\pi(k)$ being the player assigned to slot $k$.  
Being assigned to the
$k$-th slot results in $\alpha_k$ clicks, where $\alpha_1 \geq \alpha_2 \geq \hdots \geq \alpha_n$. 
Each advertiser $i$ has a private
type $v_i$, representing his or her value per click received.  The sequence
$\vals = (\val_1, \dotsc, \val_n)$ is referred to as the \emph{type profile}.
We will write $\valsmi$ for $\vals$ excluding the $i$th entry, so that
$\vals = (\vali, \valsmi)$.

A mechanism for this auction 
elicits a bid $\bidi \in [0,\infty)$ from each agent $i$, which is interpreted
as a type declaration, and returns an assignment as well as a price
$\pricei$ per click for each agent.  If advertiser $i$ is assigned to slot
$j$, his \emph{utility} is $\alpha_j (\vali - \pricei)$, which is the number
of clicks received times profit per click.  The \emph{social welfare}
of outcome $\pi$ is $SW(\pi,\vals) = \sum_j \alpha_j \val_{\pi(j)}$, the total value
of the solution for the participants.  The optimal social welfare is
$OPT(\vals) = \max_\pi SW(\pi,\vals)$.  

We focus on a particular mechanism, the Generalized Second Price auction, 
which works as follows.
Given bid profile $\bids$, the auction sets $\pi(k)$ to be
the advertiser with the $k$th highest bid (breaking ties arbitrarily).  
That is, GSP assigns slots with higher click-through-rate to agents with higher bids.  
Payments are then set according to $\pricei = \bid_{\pi(\pi^{-1}(i)+1)}$.  That is, the payment
of the $k$th highest bidder is precisely the bid of the next-highest bidder (where
we take $\bid_{n+1} = 0$).  We will write 
$\utili(\bids)$ for the utility derived by agent $i$ from the GSP
when agents bid according to $\bids$.

For the remainder of the paper, we will write $\pi(\bids,j)$ to be the player assigned
to slot $j$ by GSP when the agents bid according to $\bids$.  We will also
write $\sigma(\bids,i)$ for the slot assigned to bidder $i$ by GSP, again when agents
bid according to $\bids$.  We write $\pi^i(\bidsmi,j)$ to be the player that would
be assigned to slot $j$ by GSP if agent $i$ did not participate in the auction.
We will write $\nu(\vals)$ for the optimal assignment of slots to bidders for
value profile $\vals$, so that $\nu(\vals,i)$ is the slot that would be allocated to 
agent $i$ in the optimal assignment\footnote{We note that, since GSP makes the
optimal assignment for a given bid declaration, we actually have that $\nu(\vals,i)$
and $\sigma(\vals,i)$ are identically equal.  We define $\nu$ mainly for use when
emphasizing the distinction between an optimal assignment for a value profile and
the assignment that results from a given bid profile.}.  

\subsection{Pure and Mixed Nash Equilibrium}

A \emph{(pure) strategy} for agent $i$ is a function $b_i : \R_{\geq 0} \to \R_{\geq 0}$
that maps each private value to a declared bid.  A \emph{mixed strategy}
maps a private value to a distribution over bids, corresponding to a randomized
declaration.  

We will make the standard assumption that agents apply strategies that never 
overbid. Thatis, we restrict our attention to strategies in which $b_i(\vali)$ 
assigns probability $0$ to all bids larger than $\vali$, for all $i$ and $\vali$.  
This assumption is motivated by the fact that overbidding is a dominated strategy:
an agent's expected utility can only increase by replacing a bid larger than $\vali$
with a bid of $\vali$.

Given a value profile $\vals$, we say that strategy profile
$\bids$ is a mixed Nash equilibrium if, for all $i$ and all alternative strategies $b_i'(\cdot)$,
\[
\E u_i(b_i(v_i), b_{-i}(v_{-i})) \geq \E u_i(b_i'(v_i), b_{-i}(v_{-i})).
\]
That is, each agent $i$ maximizes his utility by bidding according to strategy $\bid_i'(\cdot)$.  We say this
is a pure Nash equilibrium if, in addition, all strategies are pure.
We define the \emph{(mixed) Price of Anarchy} to be the worst-case ratio between social welfare in the optimum and expected social welfare in GSP across all valuation profiles and all mixed Nash equilibria:
$$\sup_{\vals, \bids(\cdot) NE} \frac{OPT(\vals)}{\E_{b}[SW(\pi(\bids(\vals)),\vals)]}.$$

\subsection{Bayesian setting}

In a Bayesian setting, we suppose that each agent's type is drawn from
a publicly known distribution.  That is, $\vals \sim \dists$ where
$\mathbf{F} = F_1 \times F_2 \times \dotsc \times F_n$.  We then say that
strategy profile $\mathbf{b}$ is a Bayes-Nash equilibrium for distributions
$\mathbf{F}$ if, for all $i$, all $v_i$, and all alternative strategies $b_i'$,
\[
\E_{\valsmi \sim \distsmi}[\utili(\bidi(\vali), \bidsmi(\valsmi))]
\geq 
\E_{\valsmi \sim \distsmi}[\utili(\bid'_i(\vali), \bidsmi(\valsmi))]
\]
That is, each agent maximizes his expected utility by bidding in accordance with strategy $b_i(\cdot)$,
where expectation is taken over the distribution of the other agents' types and any randomness in
their strategies.
We define the \emph{Bayes-Nash Price of Anarchy} to be the worst-case ratio between social welfare in the optimum and social welfare in GSP across all distributions and all Nash equilibria:
$$\max_{\dists, \bids(\cdot) BNE} \frac{\E_{\vals \sim \dists}[OPT(\vals)]}{\E_{\vals \sim \dists, \bids(\vals)}[SW(\pi(\bids(\vals)),\vals]}.$$

\subsection{Repeated Auctions}


We now turn to repeated versions of GSP.  In this setting, the GSP auction is run $T \geq 1$ times with the same slots and agents.  The private value profile $\mathbf{v}$ of the agents does not change between rounds, but the agents are free to change their bids.  We write $b_i^t$ for the bid of agent $i$ on round $t$.  We refer to $D = (b^1, \dotsc, b^T)$ as a \emph{declaration sequence}.  We will write $\pi(D)$ for the sequence of permutations generated by GSP on input sequence $D$.  The average social welfare generated by GSP is then $SW(\pi(D),\vals) = \frac{1}{T}\sum_t SW(\pi(\bids^t),\vals).$

Declaration sequence $D = (b^1, \dotsc, b^T)$ \emph{minimizes external regret} for agent $i$ if, for any fixed declaration $b_i$, $\sum_t u_i(b_i^t,b_{-i}^t) \geq \sum_t u_i(b_i,b_{-i}^t) + o(T)$.  That is, as $T$ grows large, the utility of agent $i$ approaches the utility of the optimal fixed strategy in hindsight.  The \emph{Price of Total Anarchy} is the worst-case ratio between social welfare in the optimum and the average social welfare obtained by GSP across all declaration sequences that minimize external regret for all agents.  That is, the price of total anarchy is
$$\lim_{T \to \infty} \max_{\vals, D} \frac{OPT(\vals)}{SW(\pi(D),\vals)}$$
where the maximum is taken over declaration sequences that minimize external regret for all agents.



\section{Bayesian Price of Anarchy}

In this section we prove the following upper bound on the Bayesian Price of Anarchy for GSP.

\begin{theorem}
\label{thm.bpoa}
The Bayesian Price of Anarchy of GSP is at most $2(1-1/e)^{-1} \approx 3.164$.
\end{theorem}

The proof of Theorem \ref{thm.bpoa} proceeds in two steps.  We first show that a structural property of bidding profiles implies a bound on the social welfare obtained by GSP (Lemma \ref{lemma1}).  We then show that this structural property holds at all BNE of the GSP (Lemma \ref{lemma2}).

\begin{lemma}\label{lemma1}
Suppose that $\vals \sim \dists$ and agents apply strategy profile $\bids(\cdot)$.  Suppose further that the following is true:
\begin{equation}
\label{eq.struct.prop}
\E_{\valsmi} [ \alpha_{\sigma(\bids(\vals),i)} v_i + \alpha_k \bid_{\pi^i(\bidsmi(\valsmi),k)}] \geq \gamma \alpha_k v_i \quad \mbox{\text{for all slots $k$, players $i$, and values $v_i$.}} 
\end{equation}
Then $\E_{\vals \sim \dists}[SW(\pi(\bids(\vals)),\vals)] \geq \frac{1}{2}\gamma\E_{\vals \sim \dists}[OPT(\vals)]$.
\end{lemma}

\begin{lemma}\label{lemma2}
At any BNE of GSP, \eqref{eq.struct.prop} holds with $\gamma = 1 - \frac{1}{e}$.
\end{lemma}

Lemma \ref{lemma1} and Lemma \ref{lemma2} immediately imply Theorem \ref{thm.bpoa}.

\begin{proofof}{Lemma \ref{lemma1}}
Fix some value profile $\vals$.  For notational convenience, let $\Gamma$ be the induced distribution on bid profiles $\bids = \bids(\vals)$ when $\vals \sim \dists$.  Then for any player $i$, value $\vali$, and slot $k$, if we write $\bidi = \bidi(\vali)$, then we can express \eqref{eq.struct.prop} as:
\[
\E_{\bidsmi \sim \Gamma_{-i}} [\alpha_{\sigma(\bids,i)} \vali] + \E_{\bidsmi \sim \Gamma_{-i}} [ \alpha_k \bid_{\pi^i(\bidsmi,k)}] \geq \gamma \alpha_k \vali. 
\]
Note that $\valsmi$ does not appear in this expression; bids $\bidsmi$ are taken to be drawn from induced distribution $\Gamma_{-i}$.  Now, recalling that $\nu(\vals,i)$ is the slot assigned to player $i$ in the optimal assignment for values $\vals$, we can take $k = \nu(\vals,i)$ in the above inequality.  We then have
\[
\E_{\bidsmi \sim \Gamma_{-i}} [\alpha_{\sigma(\bids,i)} \vali] + \E_{\bidsmi \sim \Gamma_{-i}} [ \alpha_{\nu(\vals,i)} \bid_{\pi^i(\bidsmi,\nu(\vals,i))}] \geq \gamma \alpha_{\nu(\vals,i)} \vali
\]
for all $\vals$ and all $i$.  Notice that the strategy $\bidi(\cdot)$ does not appear in the second term, so we can rewrite as
\[
\E_{\bidsmi \sim \Gamma_{-i}} [\alpha_{\sigma(\bids,i)} \vali] + \E_{\bids \sim \Gamma} [ \alpha_{\nu(\vals,i)} \bid_{\pi^i(\bidsmi,\nu(\vals,i))}] \geq \gamma \alpha_{\nu(\vals,i)} \vali.
\]
Summing over all $i$ and taking expectation over $\vals$:
$$ \E_{\vals \sim \dists}  \left[ \sum_i \E_{\bidsmi \sim \Gamma_{-i}}[ \alpha_{\sigma(\bids,i)} \vali] \right] +  \E_{\vals \sim \dists} \left[ \sum_i \E_{\bids \sim \Gamma}[\alpha_{\nu(\vals,i)}  b_{\pi^i(\bids,\nu(\vals,i))}] \right] \geq \gamma \E_{\vals \sim \dists} \left[ \sum_i \alpha_{\nu(\vals,i)} \vali \right] $$
Consider each of the three expectations in the above expression.  For the third term, we note 
$$\E_{\vals \sim \dists}\left[ \sum_i \alpha_{\nu(\vals,i)} \vali \right] = \E_{\vals \sim \dists}[OPT(\vals)].$$ 
For the first term, linearity of expectation implies
$$\E_{\vals \sim \dists}\left[ \sum_i \E_{\bidsmi \sim \Gamma_{-i}}[  \alpha_{\sigma(\bids,i)}\vali ]\right]
= \E_{\vals \sim \dists}\left[\sum_i \alpha_{\sigma(\bids(\vals),i)}\vali \right]
= \E_{\vals \sim \dists}[SW(\pi(\bids(\vals)),\vals)].$$
For the second expectation, notice that:
$$ \E_{\vals \sim \dists, \bids \sim \Gamma} \left[ \sum_i \alpha_{\nu(\vals,i)}  b_{\pi^i(\bidsmi,\nu(\vals,i))} \right] \leq \E_{\vals \sim \dists} \E_{\bids \sim \Gamma} \left[ \sum_i \alpha_{\nu(\vals,i)}  b_{\pi(\bids,\nu(\vals,i))} \right] =  \E_{\bids \sim \Gamma} \left[ \sum_k \alpha_{k}  b_{\pi(b,k)} \right] $$
which is $\E_{\vals \sim \dists}[SW(\pi(\vals),\vals)]$.

We therefore conclude $2\E_{\vals \sim \dists}[SW(\pi(\vals),\vals)] \geq \gamma \E_{\vals \sim \dists}[OPT(\vals)]$, completing the proof.
\end{proofof}

\begin{proofof}{Lemma \ref{lemma2}}
We wish to show that $\E_{\valsmi} [ \alpha_{\sigma(\bids(\vals),i)} v_i + \alpha_k \bid_{\pi^i(\bidsmi(\valsmi),k)}] \geq \gamma \alpha_k v_i$, for all slots $k$, players $i$, and values $v_i$.  First, note that if $\alpha_k < \E_{\valsmi}[\alpha_{\sigma(\bids(\vals),i)}]$ then the result is trivial. So, let's consider $\alpha_k \geq \E_{\valsmi}[\alpha_{\sigma(\bids(\vals),i)}]$.  We'll prove that 
\[
\E_{\valsmi} [ \alpha_{\sigma(\bids(\vals),i)} v_i + \alpha_k \bid_{\pi^i(\bidsmi(\valsmi),k)}] \geq \alpha_k v_i - \E_{\valsmi} [ \alpha_{\sigma(\bids(\vals),i)} v_i]  \cdot \log \frac{ \alpha_k v_i }{\E_{\valsmi} [ \alpha_{\sigma(\bids(\vals),i)} v_i] }
\]
and then dividing everything by $\alpha_k v_i$ and using that $\frac{\log(x)}{x} \leq \frac{1}{e}$ we get the desired result.

Consider any bid $b'_i$ for agent $i$ such that $b_i < b'_i < v_i$. Then for each slot $k$, since bid $\bidi$ is utility-maximizing for agent $i$, the utility of bidding $b'_i$ satisfies
\[
\E_{\valsmi}[\alpha_{\sigma(\bids(\vals),i)} v_i] \geq \E_{\valsmi}[u_i(b_i',\bidsmi)].
\]
Also, if agent $i$ bids $b_i'$ and moreover it is true that $b_{\pi^i(\bidsmi,k)} < b_i'$, then agent $i$ will win a slot with at least $\alpha_k$ clicks.  Thus,
\[
\E_{\valsmi}[u_i(b_i',\bidsmi)] \geq (v_i - b_i')\alpha_k \prob_{\valsmi}[b_{\pi^i(\bidsmi(\valsmi),k)} < b_i'].
\]
Combining these two inequalities and substituting $z = v_i - b'_i$, we get
\[
\prob_{\valsmi}[\vali - b_{\pi^i(\bidsmi(\valsmi),k)} > z] \leq \frac{1}{z \cdot \alpha_k} \E_{\valsmi} [\alpha_{\sigma(\bids(\vals),i)} v_i].
\]
We are now able to estimate the expected value of $v_i - b_{\pi^i(\bidsmi,k)}$ using the fact that 
\[ \E_{\valsmi}[v_i - b_{\pi^i(\bidsmi(\valsmi),k)}] \geq \int_0^{\infty} \prob_{\valsmi}[v_i - b_{\pi^i(\bidsmi(\valsmi),k)} > z]dz.
\]  
Since $\vali - b_{\pi^i(\bidsmi(\valsmi),k)} \leq \vali$ with probability $1$, we have
$$\begin{aligned}
\E[v_i - b_{\pi^i(\bidsmi,k)}] 
& = \int_{0}^{v_i}\prob_{\valsmi}[\vali - b_{\pi^i(\bidsmi(\valsmi),k)} > z] dz \\
& \leq \int_{0}^{\E_{v_{-i}} [\alpha_{\sigma(\bids(\vals),i)} v_i] / \alpha_k } 1 dz + \int_{\E_{v_{-i}} [\alpha_{\sigma(\bids(\vals),i)} v_i] / \alpha_k }^{v_i }\frac{ \E_{v_{-i}} [\alpha_{\sigma(\bids(\vals),i)} v_i]}{\alpha_k z} dz \\
& \leq \frac{ \E_{v_{-i}} [\alpha_{\sigma(\bids(\vals),i)} v_i]} {\alpha_k} + \frac{\E_{v_{-i}} [\alpha_{\sigma(\bids(\vals),i)} v_i]}{\alpha_k}\left(\log \vali - \log \frac{\E_{v_{-i}} [\alpha_{\sigma(\bids(\vals),i)} v_i] }{ \alpha_k}  \right).
\end{aligned}$$
Multiplying both sides by $\alpha_k$ and rearranging gives the required inequality.
\end{proofof}

\subsection{Correlated bids and Price of Total Anarchy}
\label{sec.pota}

Notice that the proof of the previous section applies even in cases where agent bids are coarsely correlated. In such a case, we can consider a common source of randomness $\mathcal{R}$ and each bidding function to be a function $b_i(v_i, r)$, where $r \sim \mathcal{R}$. We call a profile of bidding functions a coarse correlated equilibrium if:
\[
\E_{\valsmi \sim \distsmi}[\utili(\bidi(\vali,r), \bidsmi(\valsmi,r))]
\geq 
\E_{\valsmi \sim \distsmi}[\utili(\bid'_i(\vali,r), \bidsmi(\valsmi,r))], \forall i, v_i, r
\]

We still suppose $v_i \sim F_i$ where $F_i$ are independent distributions. In this case, $F$ and $\mathcal{R}$ induce a distribution $\Gamma$ on the bids.

Adapting Lemma \ref{lemma2} to this context is straightforward. Now, to adapt Lemma \ref{lemma1}, observe that the only additional requirement is that we must argue that

$$ \E_{\bidsmi \sim \Gamma_{-i}} [ \alpha_{\sigma(v,i)} b_{\pi^i(\bidsmi,\sigma(v,i))}] =  \E_{\bids \sim \Gamma} [ \alpha_{\sigma(v,i)} b_{\pi^i(\bidsmi,\sigma(v,i))}].  $$

However, this follows because the marginal of $b \sim \Gamma$ restricted to $-i$ is exactly $b_{-i} \sim \Gamma_{-i}$.

We now note that the result from the previous section implies a bound on the price of total anarchy for GSP.  This follows because, whenever bidding sequence $D = (\bids^1, \dotsc, \bids^T)$ minimizes regret for all agents, the bidding strategy with shared randomness $\bidi(\vali,t) = \bidi^t(\vali)$ for $t \in [T]$ is a coarse correlated equilibrium.  Lemmas \ref{lemma1} and \ref{lemma2} therefore imply that, for all $\vals$, 
$$\E_{t \in [T]}[SW(\pi(\bids^t(\vals)),\vals)] \geq \frac{1}{2}(1-1/e) OPT(\vals)$$
which implies that the price of total anarchy is bounded by $2(1-1/e)^{-1}$.

\subsection{Byzantine Agents}

We now consider a setting in which, in addition to the $n$ advertisers who bid rationally, there are $m$ ``byzantine'' advertisers who may bid irrationally.  Write $N$ for the set of rational advertisers, and $M$ for the set of irrational advertisers.  Note that we still think of the irrational advertisers as being true players, who stil receive value per click.  The irrational bidders simply cannot be assumed to play at equilibrium; for example, they may not have experience with the GSP auction, or not know about historical bidding patterns.  

Given an outcome $\pi$ (which is an assignment of these $n+m$ bidders to $n+m$ slots), the definition of social welfare is unchanged: it is $SW(\pi,\vals) = \sum_{i \in N \cup M} \vali \alpha_{\pi(i)}$.  We define the social welfare of bidders in $N$ to be precisely that: $SW_N(\pi,\vals) = \sum_{i \in N} \vali \alpha_{\pi(i)}$.  The optimal social welfare for bidders in $N$ is $OPT_N(\vals) = \max_{\pi} SW_N(\pi,\vals)$.

We wish to show that the total social welfare obtained by GSP is a good approximation to $OPT_N(\vals)$ when the players in $N$ play at equilibrium and the players in $M$ play arbitrarily.  That is, the addition of irrational players does not degrade the social welfare guarantees of GSP had they not participated.  In order to make this claim, we must impose a restriction on the behaviour of the irrational players: that they do not overbid.  In other words, we require that $\bidi(\vali) \leq \vali$ for all $i \in M$ and all $\vali$.  We feel this is a natural restriction: overbidding is easily seen to be a dominated strategy (i.e.\ any strategy that bids higher than $\vali$ is dominated by a strategy that lowers such bids to be at most $\vali$).  Moreover, it is arguable that inexperienced bidders would bid conservatively, and not risk a large payment with no gain.

Given that byzantine agents do not overbid, we note that our BPoA bounds go through in this setting almost without change.  In particular, our structural property \eqref{eq.struct.prop} continues to hold for all agents in $N$.

\begin{lemma}\label{lemma2-byz}
Equation \eqref{eq.struct.prop} holds with $\gamma = 1 - \frac{1}{e}$ for all for all $i \in N$.
\end{lemma}
\begin{proof}
This proof follows the proof of Lemma \ref{lemma2} without change.  Note that in that proof we used only the fact that the bidding strategy of agent $i$ is a best response, so the fact that other agents may not bid at equilibrium does not affect the argument.
\end{proof}

The corresponding version of Lemma \ref{lemma1} then follows from \eqref{eq.struct.prop} just as in the setting without byzantine agents.

\begin{lemma}\label{lemma1-byz}
If \eqref{eq.struct.prop} holds for all $i \in N$, then $\E_{\vals \sim \dists}[SW(\pi(\bids(\vals)),\vals)] \geq \frac{1}{2}\gamma\E_{\vals \sim \dists}[OPT_N(\vals)]$.
\end{lemma}
\begin{proof}
Precisely as in the proof of Lemma \ref{lemma1}, we obtain
\[
\E_{\vals \sim \dists}  \left[ \sum_{i \in N} \E_{\bidsmi \sim \Gamma_{-i}}[ \alpha_{\sigma(\bids,i)} \vali] \right] +  \E_{\vals \sim \dists} \left[ \sum_{i \in N} \E_{\bids \sim \Gamma}[\alpha_{\nu(\vals,i)}  b_{\pi^i(\bids,\nu(\vals,i))}] \right] \geq \gamma \E_{\vals \sim \dists} \left[ \sum_{i \in N} \alpha_{\nu(\vals,i)} \vali \right]
\]
where we note that the summations are over agents in $N$. Then, as in Lemma \ref{lemma1},
$$\E_{\vals \sim \dists}\left[ \sum_{i \in N} \alpha_{\nu(\vals,i)} \vali \right] = \E_{\vals \sim \dists}[OPT(\vals)].$$ 
and
$$\E_{\vals \sim \dists}\left[ \sum_{i \in N} \E_{\bidsmi \sim \Gamma_{-i}}[  \alpha_{\sigma(\bids,i)}\vali ]\right]
= \E_{\vals \sim \dists}[SW_N(\pi(\bids(\vals)),\vals)] \leq \E_{\vals \sim \dists}[SW(\pi(\bids(\vals)),\vals)].$$
For the second expectation, notice that:
\[
\begin{aligned}
& \E_{\vals \sim \dists, \bids \sim \Gamma} \left[ \sum_{i \in N} \alpha_{\nu(\vals,i)}  b_{\pi^i(\bidsmi,\nu(\vals,i))} \right] 
\leq \E_{\vals \sim \dists} \E_{\bids \sim \Gamma} \left[ \sum_{i \in N} \alpha_{\nu(\vals,i)}  b_{\pi(\bids,\nu(\vals,i))} \right] \\
& = \E_{\bids \sim \Gamma} \left[ \sum_{k : \exists i \in N, \nu(\vals,i)=k}  \alpha_{k}  b_{\pi(\bids,k)} \right] 
 \leq \E_{\bids \sim \Gamma} \left[ \sum_k \alpha_{k}  \val_{\pi(\bids,k)} \right]  
= \E_{\vals \sim \dists} [SW(\pi(\vals),\vals)]. 
\end{aligned}
\]

We therefore conclude $2\E_{\vals \sim \dists}[SW(\pi(\vals),\vals)] \geq \gamma \E_{\vals \sim \dists}[OPT_N(\vals)]$, completing the proof.
\end{proof}



% ...

Together, Lemma \ref{lemma2-byz} and Lemma \ref{lemma1-byz} imply that the Bayesian Price of Anarchy of GSP is at most $2(1-1/e)^{-1}$ even in the presence of irrational bidders.  Following the comments in Section \ref{sec.pota}, we can apply the same argument to obtain a matching bound on the Price of Total Anarchy with irrational bidders.

\section{Towards a Tight Pure PoA}

In \cite{PLT10}, Paes Leme and Tardos give a bound of 1.618 for the
Pure Price of Anachy of GSP. They also prove that for $n = 2$ slots,
the correct bound is $1.25$. Here we show that for $n = 3$, the
correct bound is $1.25913$. We conjecture that this is the correct
Price of Anarchy for GSP (for any number of slots) and we suggest an
approach to prove this result.

\begin{lemma}\label{3-slot-lemma}
       For $n = 3$ slots, the pure Price of Anarchy of GSP is exactly $1.25913$.
\end{lemma}

\begin{proof}
Fix one permutation $\pi$. If there is an $i$ s.t. $\pi(i) = i$ then
it is easy to show the Price of Anarchy is bounded by $1.25$. This
excludes all but two allocations which we analyze below. They are: (i)
$\pi = [2,3,1]$ and (ii) $\pi = [3,1,2]$.

\emph{Case (i):} $\pi = [2,3,1]$. We can write the price of anarchy as:
$$PoA = \frac{\alpha_1 v_1 + \alpha_2 v_2 + \alpha_3 v_3}{\alpha_3 v_1
+ \alpha_1 v_2 + \alpha_2 v_3}$$
Now, notice that the coefficient of $v_2$ is smaller in the numerator
than in the denominator. The same is true for $v_3$. Now, we use the
following simple observation about ratios: if $a \leq b$ and $v \geq
v'$ then: $\frac{x+av}{y+bv} \leq \frac{x+av'}{y+bv'}$, which is
natural, because decreasing $v$ we decrease the denominator more than
the numerator. Now, we use that technique to bound $v_2$ and $v_3$ in
terms of $v_1$:

\begin{itemize}
 \item $v_2 \geq \frac{\alpha_1 - \alpha_3}{\alpha_1} v_1$
 \item $v_3 \geq \frac{\alpha_2 - \alpha_3}{\alpha_2}v_1$
\end{itemize}
The first inequality comes from the Nash inequalities $\alpha_3(v_1 -
0) \geq \alpha_1(v_1 - b_2) \geq \alpha_1(v_1 - v_2)$ and the second
comes from the fact that $\alpha_3(v_1 - 0) \geq \alpha_2(v_1 - b_3)
\geq  \alpha_2(v_1 - v_3)$. Now, we get:

\begin{equation}\label{3-slot-eqn}
PoA \leq \frac{\alpha_1 v_1 + \alpha_2 \left[ \frac{\alpha_1 -
\alpha_3}{\alpha_1} v_1 \right] + \alpha_3 \left[ \frac{\alpha_2 -
\alpha_3}{\alpha_2}v_1 \right] }{\alpha_3 v_1 + \alpha_1
\left[\frac{\alpha_1 - \alpha_3}{\alpha_1}  v_1 \right] + \alpha_2
\left[ \frac{\alpha_2 - \alpha_3}{\alpha_2}v_1 \right] }
\end{equation}

Which allows us to eliminate $v_1$ and optimize for $\alpha$. By
standard techniques one can prove that the optimum is $1.25913$ which
is the root of a fourth degree equation. The values for which it is
maximized are $\alpha_1 = 1, \alpha_2 = 0.55079, \alpha_3 = 0.4704$.

\emph{Case (ii):} $\pi = [3,1,2]$. We can write the price of anarchy as:
$$PoA = \frac{\alpha_1 v_1 + \alpha_2 v_2 + \alpha_3 v_3}{\alpha_2 v_1
+ \alpha_3 v_2 + \alpha_1 v_3}$$
and again we use the same trick of realizing that $v_1 \leq
\frac{\alpha_1}{\alpha_1 - \alpha_2} v_3$ by the fact that player $1$
doesn't want to get the first slot, and $v_2 \leq
\frac{\alpha_1}{\alpha_1 - \alpha_3} v_3$ by the fact that player $2$
doesn't want to take the first slot. That gives us:

$$PoA \leq \frac{\alpha_1 \left[ \frac{\alpha_1}{\alpha_1 - \alpha_2}
v_3 \right] + \alpha_2 \left[  \frac{\alpha_1}{\alpha_1 - \alpha_3}
v_3 \right] + \alpha_3 v_3}{\alpha_2 \left[ \frac{\alpha_1}{\alpha_1 -
\alpha_2} v_3 \right] + \alpha_3 \left[  \frac{\alpha_1}{\alpha_1 -
\alpha_3} v_3 \right] + \alpha_1 v_3}$$
which has the same solution $1.25913$ when maximized. Now, it is
maximized for $\alpha_1 = 1, \alpha_2 = 0.5295, \alpha_3 = 0.1458$. In
fact, it is not hard to see that those two PoA expressions have the
same maximum: given a point $(1,\alpha_2, \alpha_3)$ (wlog we can
consider $\alpha_1 = 1$ because the expression is homogeneous), the
second expressions evaluates to the same value in the point
$(1,1-\alpha_3,\frac{\alpha_2-\alpha_3}{\alpha_2})$.
\end{proof}

We proved that $1.25913$ is the tight Price of Anarchy for $3$ slots
(we can use the optimization results in Case(i) to generate a tight
example). We also conjecture that this is the correct Price of Anarchy
for any $n\geq 3$. Moreover, we conjecture that the allocation
maximizing the Price of Anarchy for $n$ slots is $\pi =
[2,3,4,\hdots,n,1]$, i.e., the player with higher value takes the
bottom slot and all players $i > 1$ take slot $i-1$. Then, if this is
the case, we can prove our desired theorem by showing the following
result:

\begin{lemma}
 If an equilibrium with $n$ players and $n$ slots is such that
$\sigma(1) = n$ and $\sigma(i) = i-1$ for the other players, then the
Price of Anarchy is $1.25913$.
\end{lemma}

\begin{proof}
 Following a proof scheme similar to used in Lemma \ref{3-slot-lemma}
we can write:
$$PoA = \frac{\alpha_1 v_1 + \sum_{i>1} \alpha_i v_i}{\alpha_n v_1 +
\sum_{i>1} \alpha_{i-1} v_i} \leq \frac{\alpha_1 v_1 + \sum_{i>1}
\alpha_i \left[ \frac{\alpha_{i-1} - \alpha_n}{\alpha_{i-1}} v_1
\right]}{\alpha_n v_1 + \sum_{i>1} \alpha_{i-1} \left[
\frac{\alpha_{i-1} - \alpha_n}{\alpha_{i-1}} v_1 \right]} $$
This boils down to optimizing a function on multiple variables. It can
be shown using standard techniques from optimization that the optimum
is the same of equation \ref{3-slot-eqn}. In fact, if $(\alpha_1,
\alpha_2, 1)$ is a solution to $3$ slots, then $(\alpha_1, \alpha_2,
1, \hdots, 1)$ is a solution for $n$ slots.
\end{proof}

%

\bibliography{sigproc}
\bibliographystyle{plain}

\end{document}